\documentclass[11pt]{article}
\usepackage{fullpage,latexsym,amsthm,amsfonts,amssymb,amsmath,amsthm}
\usepackage{tikz}

\def\01{\{0,1\}}
\newcommand{\ceil}[1]{\lceil{#1}\rceil}

\newcommand{\ket}[1]{|#1\rangle}
\newcommand{\bra}[1]{\langle#1|}
\newcommand{\ketbra}[2]{|#1\rangle\langle#2|}
\newcommand{\norm}[1]{{\left\|{#1}\right\|}}

\newcommand{\rank}{\mbox{\rm rank}}
\newcommand{\nrank}{\mbox{\rm nrank}}
\newcommand{\fool}{\mbox{\rm fool}}

\newcommand{\supp}{\mbox{\rm supp}}
\newcommand{\DISJ}{\mbox{\sc Disj}}
\newcommand{\EQ}{\mbox{\sc EQ}}
\newcommand{\GT}{\mbox{\sc GT}}

\newtheorem{theorem}{Theorem}
\newtheorem{lemma}{Lemma}
\newtheorem{corollary}{Corollary}
\newtheorem{claim}{Claim}

\begin{document}

\title{Fooling One-Sided Quantum Protocols}
\author{Hartmut Klauck\thanks{Research at the Centre
for Quantum Technologies is funded by the Singapore Ministry of Education
and the National Research Foundation.}
\\
CQT and Nanyang Technological University\\ Singapore\\
{\tt hklauck@gmail.com}
\and
Ronald de~Wolf\thanks{Partially supported by a Vidi grant from the Netherlands Organization for Scientific Research (NWO), and by the European Commission under the project QCS (Grant No.~255961).
Most of this work was done when RdW was visiting CQT, whose hospitality is gratefully acknowledged.}\\
CWI and University of Amsterdam\\ The Netherlands\\
{\tt rdewolf@cwi.nl}
}
\date{}
\maketitle

\begin{abstract}
We use the venerable ``fooling set'' method to prove new lower bounds
on the quantum communication complexity of various functions.
Let $f:X\times Y\rightarrow\01$ be a Boolean function, $\fool^1(f)$ its maximal fooling set size among 1-inputs,
$Q_1^*(f)$ its one-sided-error quantum communication complexity with prior entanglement, and
$NQ(f)$ its nondeterministic quantum communication complexity (without prior entanglement;
this model is trivial with shared randomness or entanglement).
Our main results are the following, where logs are to base~2:
\begin{itemize}
\item If the maximal fooling set is ``upper triangular'' (which is for instance the case for the equality, disjointness, and greater-than functions), then we have $Q_1^*(f)\geq\frac{1}{2}\log \fool^1(f) - \frac{1}{2}$,
which is essentially optimal by superdense coding.  No super-constant lower bound for equality seems to follow from earlier techniques.
\item For all $f$ we have $Q_1^*(f)\geq\frac{1}{4}\log \fool^1(f) - \frac{1}{2}$, which is optimal up to a factor of~2.
\item $NQ(f)\geq\frac{1}{2}\log \fool^1(f) + 1$.
We do not know if the factor~1/2 is needed in this result, but
it cannot be replaced by~1: we give an example where $NQ(f)\approx 0.613 \log \fool^1(f)$.
\end{itemize}
\end{abstract}

\section{Introduction}

\subsection{Background: fooling classical communication protocols}

Communication complexity~\cite{yao:distributive,kushilevitz&nisan:cc} is one of the most versatile and successful computational models we have,
and \emph{lower bounds} on communication complexity are one of the main sources of lower bounds in many other areas,
from circuits to data structures to data streams.
One of the simplest and most intuitive ways to prove lower bounds on communication protocols is by
exhibiting a large \emph{fooling set}, which was first done in~\cite{yao:distributive,lipton&sedgewick:vlsi}.
Suppose Alice and Bob want to compute some function $f:X\times Y\rightarrow\01$, given inputs $x\in X$ and $y\in Y$, respectively.
A 1-fooling set for~$f$ is a set $F=\{(x,y)\}$ of input pairs with the following properties:
\begin{quote}
(1) If $(x,y)\in F$ then $f(x,y)=1$\\
(2) If $(x,y),(x',y')\in F$ then $f(x,y')=0$ or $f(x',y)=0$
\end{quote}
By renaming some of Bob's inputs we can always assume without loss of generality that $F$ is of the form $\{(x,x)\}$.

For example, consider the $n$-bit equality function $\EQ$,
defined on $x,y\in\01^n$ as $\EQ(x,y)=1$ iff $x=y$.
This has a 1-fooling set $F=\{(x,x)\}$ of size $2^n$, since $\EQ(x,x)=1$ for all $x$ and $\EQ(x,y)=0$ for all distinct $x,y$.
The same fooling set also works for the $n$-bit greater-than function, which is defined as $\GT(x,y)=1$ iff $y\geq x$.
The $n$-bit disjointness function $\DISJ$, defined as $\DISJ(x,y)=1$ iff $|x\wedge y|=0$,
also has a 1-fooling set of size $2^n$, which can be seen as follows:
write its communication matrix as $\left(\begin{array}{ll}1 & 1\\ 1 & 0\end{array}\right)^{\otimes n}$,
and take the anti-diagonal as the 1-fooling set.  All entries on the anti-diagonal are~1 (giving the first property)
and all entries below the anti-diagonal are~0 (giving the second property).

Now consider for simplicity a deterministic protocol computing $f$.
Suppose the last bit of the conversation is the output bit, so both parties end up knowing the output.
Consider input pairs $(x,y),(x',y')\in F$.
For both inputs, the first property of the fooling set says that the correct output value is~1.
Suppose, by way of contradiction, that the conversation between Alice and Bob is the same on both input pairs.
If we switch input pair $(x,y)$ to $(x,y')$ then nothing changes from Alice's perspective
(neither her input nor the conversation changes), so the output will still be~1.
Similarly, if we switch $(x,y)$ to $(x',y)$ then the output won't change from Bob's perspective.
But by the second property of fooling sets, for at least one of $(x,y')$ and $(x',y)$, the correct output is~0!
Hence the conversations on inputs $(x,y)$ and $(x',y')$ must have been different.
Accordingly, the bigger our fooling set $F$ is,
the more distinct conversations we must allow and hence the more bits of communication are needed.

More precisely, the communication complexity is lower bounded by $\log|F|+1$.
A formal proof of this fact can be based on the notion of \emph{monochromatic rectangles}.
A rectangle is a set $R=A\times B$, where $A\subseteq X$ and $B\subseteq Y$.
Such a rectangle is \emph{1-monochromatic} if $f(x,y)=1$ for all $(x,y)\in R$.
Note that a rectangle containing 1-inputs $(x,y),(x',y')\in F$ cannot be 1-monochromatic,
because by the rectangle property it also contains $(x,y')$ and $(x',y)$,
at least one of which is a 0-input by the fooling set property.
Accordingly, if we want to include $F$ in a set of 1-rectangles, we need a separate 1-rectangle for each
element of $F$, and hence need at least $|F|$ different rectangles.
It is well-known that a deterministic $c$-bit communication protocol induces a partition
of the set of all 1-inputs into $2^{c-1}$ 1-monochromatic rectangles, so the previous argument implies $2^{c-1}\geq|F|$;
equivalently $c\geq\log|F| +1$.
In fact even \emph{nondeterministic} communication complexity is lower bounded by $\log |F| + 1$,
since a $c$-bit nondeterministic protocol gives rise to a \emph{cover} (rather than partition) of
the set of all 1-inputs by $2^{c-1}$ 1-monochromatic rectangles, and we still need a separate rectangle for each element of $F$.

In contrast, a \emph{quantum} communication protocol does not naturally induce a partition or cover of the 1-inputs into rectangles%
\footnote{It can be viewed as approximately producing rectangles \emph{with signs}~\cite[Section~3]{klauck:qcclowerj}.},
so the above way of reasoning fails.
In fact, in contrast to the classical case, the number of monochromatic rectangles needed to partition
the 1-inputs does not provide a lower bound on exact quantum protocols,
as witnessed by the exponential separation in~\cite{BuhrmanCleveWigderson98}.
Nevertheless, in this paper we show how fooling sets can still be used to lower bound quantum communication complexity.
We do this in two settings: one-sided-error quantum protocols with unlimited prior entanglement
and nondeterministic quantum protocols without entanglement. These results also imply lower bound for quantum ``Las Vegas'' or ``zero-error'' protocols (i.e., quantum protocols that never err, but have probability $\leq 1/2$ of giving up without a result).

\subsection{Our results: fooling one-sided-error quantum protocols}

First, we study one-sided-error protocols: protocols that always output~0 on inputs $x,y$ where $f(x,y)=0$,
and that output~1 with probability at least 1/2 on inputs where $f(x,y)=1$. We start by getting an essentially optimal bound for the case of ``upper-triangular'' fooling sets.  We call a 1-fooling set $F=\{(x,x)\}$ upper-triangular if there is some total ordering `$\geq$' on the $x$'s such that $x>y$ implies $f(x,y)=0$.  In other words, the matrix $M$ with entries $M_{xy}=f(x,y)$ is~0 below the diagonal.  In Section~\ref{secqonesidedLB} we show that if $f$ has an upper-triangular 1-fooling set of size $N$, then
$$
Q_1^*(f)\geq\frac{1}{2}\log N - \frac{1}{2}.
$$
For example, the $n$-bit equality, disjointness, and greater-than functions all have upper-triangular 1-fooling sets of size $N=2^n$, and hence an $n/2-1/2$ lower bound on their one-sided-error complexity $Q_1^*(f)$.
We have $Q^*_1(f)\leq n/2+1$ for any Boolean function where $X\subseteq\01^n$, because superdense coding~\cite{superdense}
allows Alice to send 2 classical bits using one EPR-pair and one qubit of communication.
Hence the above result is essentially tight for the functions mentioned.

We can extend this to a slightly weaker result for all functions stated in terms of their (not necessarily upper-triangular) 1-fooling-set size:
$$
Q_1^*(f)\geq\frac{1}{4}\log \fool^1(f) - \frac{1}{2}.
$$
Surprisingly for such basic functions as equality and disjointness, these bounds were not known before.
While it is possible to use Razborov's technique~\cite{razborov:qdisj} combined with results about polynomial approximation
with very small error~\cite{bcwz:qerror} to show $Q^*_1(\DISJ)=\Omega(n)$, no super-constant lower bound was known for $Q^*_1(\EQ)$.
This gap in our knowledge was due to the fact that other existing lower bound methods cannot give good lower bounds for equality,
as we explain now.  General lower bound methods for quantum communication complexity
can be grouped into rank-based methods and methods based on approximation norms (in particular based on
the $\gamma_2$-norm \cite{linial&shraibman:ccj}).\footnote{Information-theoretic methods~\cite{jrs:disjointness} have also been used to lower bound quantum communication complexity. However, the notion is defined there for internal information cost, and in this case the information cost for equality is $O(1)$, even for classical protocols without error~\cite[Proposition 3.21]{Brav:interactive}.}
The linearity of norms makes it possible to prove lower bounds for quantum
protocols in which Alice and Bob share prior entanglement.
Rank-based methods, however, do not seem to directly apply to protocols with entanglement: in the case of exact quantum protocols a direct sum-based construction in~\cite{buhrman&wolf:qcclower}
shows that the logarithm of the rank is a lower bound even in the presence of entanglement.\footnote{Footnote~2
of~\cite{buhrman&wolf:qcclower} claims such a bound for \emph{zero-error} quantum protocols for equality and disjointness without proof,
but in retrospect they didn't seem to have a proof of this.}
In the case of two-sided error and entanglement, Lee and Shraibman \cite{lee&shraibman:apprank} show
that the approximation rank yields lower bounds by relating it to the $\gamma_2$-norm.
Since the communication matrix of $\EQ$ is the identity matrix $I$, and $\gamma_2(I)=O(1)$ for $I$ of any size, there is no hope to use a connection between a one-sided-error version of approximation rank and the $\gamma_2$-norm to establish a large lower bound on $Q^*_1(\EQ)$.
Whether a one-sided-error version of approximation rank gives lower bounds for $Q^*_1$ remains open,
but we note that the construction in~\cite{lee&shraibman:apprank} cannot be adapted to the one-sided-error scenario.

So neither of the two main approaches to quantum communication complexity
lower bounds provides us with a good lower bound for $Q^*_1(\EQ)$. Hence in this paper we take a different approach. We first simulate a quantum protocol with entanglement by a game without communication, in which Alice and Bob share entanglement, and they need to compute a function $f$ conditioned on postselection on their local measurements. This approach itself is not new, and can for instance be used to show that the $\gamma_2$-norm is a lower bound, see~\cite{lee&shraibman:survey}. We then analyze the impact of Alice and Bob's measurements on the single entangled state used in the game. The one-sided-error requirement places strong constraints on those measurements, which we exploit to derive our lower bound in terms of fooling sets.

In a quantum \emph{Las Vegas} protocol Alice and Bob compute a function $f$ without error, but they are allowed to give up without a result with probability 1/2. The quantum Las Vegas communication complexity with entanglement $Q^*_0(f)$ is the minimum worst-case communication of any protocol that computes $f$ under these requirements.\footnote{It is possible to define Las Vegas protocols as protocols that never err and place bounds on \emph{expected} communication. The corresponding complexity measure is always larger or equal to the one considered here, and is smaller than 2 times our measure.}
Quantum Las Vegas protocols were investigated in~\cite{bcwz:qerror,klauck:qpcom,wolf:nqj} in the case where no prior entanglement is available.
Since $Q^*_0(f)\geq\max\{Q^*_1(f), Q^*_1(\neg f)\}$ we immediately get large lower bounds
on the quantum Las Vegas complexity of $\DISJ$ and $\EQ$, and also the following general lower bound:
$$
Q_0^*(f)\geq\frac{1}{4}\log \fool(f) - \frac{1}{2},
$$
where $\fool(f)$ is the standard maximum fooling set size, i.e., the maximum over the largest 1-fooling set and 0-fooling set.

\subsection{Our results: fooling nondeterministic quantum protocols}

As a second main result, just like in the classical world fooling sets lower bound \emph{nondeterministic} protocols,
we show here that they also lower bound nondeterministic \emph{quantum} protocols.
For our purposes, we can define a nondeterministic protocol (quantum as well as classical)
for a Boolean function $f$ as one that has positive acceptance probability on input $x,y$ iff $f(x,y)=1$.
In other words, this is the unbounded-error version of the one-sided-error model: the requirement of acceptance probability~0 on 0-inputs remains, but the requirement of \emph{large} acceptance probability on 1-inputs is relaxed to \emph{positive} acceptance probability on 1-inputs.%
\footnote{Nondeterministic communication complexity (classical as well as quantum) can be exponentially less than one-sided-error communication complexity, even if the latter is assisted by unlimited prior entanglement.  The negation of the disjointness function is an example of this.}
The quantum version of this model was introduced in~\cite{wolf:nqj}, which also exhibits a total function
with an exponential separation between quantum and classical nondeterministic communication complexities.

Note that allowing unlimited prior entanglement trivializes the nondeterministic model,
for the same reason that unlimited shared randomness trivializes it in the classical case:
Alice and Bob can share a random variable $r$ uniformly distributed over the set $X$ of Alice's
inputs; Alice sends a bit indicating whether $x=r$; if `yes' then Bob outputs $f(r,y)=f(x,y)$,
and if `no' then he outputs~0. Hence if we were to allow unlimited prior randomness or entanglement,
any function has nondeterministic communication complexity at most~1.
Accordingly, we will study nondeterministic protocols which don't share anything at the start.
In Section~\ref{secqndetLB} we show the following lower bound on nondeterministic quantum communication
complexity in terms of fooling~sets:
$$
NQ(f)\geq\frac{1}{2}\log \fool^1(f) + 1.
$$
We do not know if the factor~1/2 is needed in this result, but it cannot be replaced by~1:
in Section~\ref{secqndetLB} we give an example of a function where $NQ(f)\leq \frac{\log 3}{\log 6}\log \fool^1(f)+1$,
where $\log 3/\log 6\approx0.613$.

\section{Lower bound for one-sided bounded-error quantum protocols}\label{secqonesidedLB}

We assume familiarity with communication complexity.
See~\cite{kushilevitz&nisan:cc} for more details about classical communication complexity
and~\cite{wolf:qccsurvey} for quantum communication complexity.
Our key lemma is based on a reasonably well-known trick to replace quantum communication
by the guessing of twice as many classical bits:

\begin{lemma}\label{lemtonocom}
Suppose there is a quantum protocol $P$ with inputs from $X\times Y$ and output in $\01$,
that uses some fixed starting state (possibly entangled) and $q$ qubits of communication,
and where a measurement of the last qubit on the channel gives the output.
Then there exists another quantum protocol $Q$ with a fixed starting state
and no communication at all, where Alice outputs $a\in\01$ and Bob outputs $b\in\01$, such that
$$
\mbox{for all inputs }x,y: \Pr[Q~\mbox{outputs }a=b=1]=2^{-2q}\Pr[P~\mbox{outputs~1}].
$$
\end{lemma}

\begin{proof}
We assume without loss of generality that $P$ communicates \emph{exactly} $q$ qubits on all possible inputs.
By the well-known teleportation primitive~\cite{teleporting},
we can replace each qubit of communication in $P$ by the use of one additional EPR-pair
and two classical bits of communication.
These 2 bits are the outcome of a measurement by the sending party,
and indicate which of the 4 Pauli matrices the receiving party has to apply on their end of
the EPR-pair in order to obtain the qubit that the sender wanted to send.
If the bits happen to be~00 (which happens with probability 1/4),
then the right Pauli is the identity matrix, so then they don't need to do anything.
Call the resulting $2q$-bit protocol~$P_{clas}$.

Protocol~$Q$ is now as follows.
Alice and Bob run protocol $P_{clas}$ \emph{assuming} all messages are 0-bits (so they don't communicate anything).
Alice checks if all her teleportation measurements gave outcome~00. If not then she outputs $a=0$;
if yes then she outputs $P_{clas}$'s output if she was the one supposed to output that,
and otherwise she outputs $a=1$.
Bob does the same from his end, outputting $b\in\01$.
Note that $a=b=1$ iff all $q$ teleportation measurements gave outcome~00 \emph{and}
the output of $P$ was~1. The first event happens with probability $4^{-q}$ and the second
event with $\Pr[P~\mbox{outputs~1}]$.
Since these two events are independent we can multiply their probabilities to obtain the lemma.
\end{proof}

Note that the starting state of the new protocol~$Q$ is the starting state of the original protocol~$P$, augmented with an additional $q$ EPR-pairs.
Using the above lemma we can prove an essentially optimal lower bound in terms of upper-triangular 1-fooling sets:

\begin{theorem}\label{thq1fool}
If $f:X\times Y\rightarrow\01$ has an upper-triangular 1-fooling set of size $N$,
then $$
Q_1^*(f)\geq\frac{1}{2}\log N - \frac{1}{2}.
$$
\end{theorem}

\begin{proof}
We can assume without loss of generality that the fooling set is of the form $\{(x,x):x\in[N]\}$, and $f(x,y)=0$ whenever $x>y$.
Let $q=Q_1^*(f)$ and let $P$ be a $q$-qubit entanglement-assisted protocol for~$f$.
Apply Lemma~\ref{lemtonocom} to this protocol to obtain a new protocol $Q$ without communication,
where Alice outputs $a\in\01$, Bob outputs $b\in\01$, satisfying
\begin{quote}
$\Pr[a=b=1]\geq 2^{-2q-1}$ on inputs $(x,x)\in F$\\
$\Pr[a=b=1]=0$ on inputs $x>y$
\end{quote}
Let $\ket{\psi}$ be the entangled starting state of protocol $Q$, which we assume to be pure without loss of generality.
On input $x$, Alice applies a POVM measurement with operators $A_x,I-A_x$ corresponding to outputs 1 and 0, respectively.  Similarly Bob uses POVM elements $B_y,I-B_y$.
The following technical claim is the core of the proof:

\begin{claim}\label{claim:key}
Let $\ket{w}$ be a bipartite state such that for all $x,y\in[N]$ satisfying $x>y$, we have $\bra{w}A_x\otimes B_y\ket{w}=0$.
Then 
$$
\sum_{x\in[N]} \bra{w}A_x\otimes B_x\ket{w} \leq \norm{w}^2.
$$
\end{claim}

\begin{proof}
The proof is by induction on $N$.
The base case $N=1$ follows from the Cauchy-Schwarz inequality and the fact that $A_x\otimes B_x$ has operator norm $\leq 1$.

For the inductive step: assume the claim holds for $N$, and now let $x$ range over $[N+1]$.
Fix some bipartite state $\ket{w}$ such that 
\begin{quote}
(*)~for all $x,y\in[N+1]$ satisfying $x>y$, we have $\bra{w}A_x\otimes B_y\ket{w}=0$. 
\end{quote}
Let $A_{N+1}=\sum_i\alpha_i\ketbra{a_i}{a_i}$, with $\alpha_i\in(0,1]$, be the spectral decomposition of POVM element $A_{N+1}$.
Let $\supp(A_{N+1})=\sum_i\ketbra{a_i}{a_i}$ denote the projection on the support of $A_{N+1}$.
Define $\ket{w_1} = (\supp(A_{N+1}) \otimes I) \ket{w}$, and $\ket{w_2}=\ket{w}-\ket{w_1}$.  
For $y\in[N]$, let $B_y=\sum_j\beta_j\ketbra{b_j}{b_j}$, with $\beta_j\in(0,1]$, be the spectral decomposition of POVM element $B_y$.
By~(*) we have 
$$
0=\bra{w}A_{N+1}\otimes B_y\ket{w}=\sum_{i,j}\alpha_i\beta_j\left|\bra{w}\cdot \ket{a_i}\otimes \ket{b_j}\right|^2.
$$
This means that $\ket{w}$ is orthogonal to all eigenvectors $\ket{a_i}\otimes\ket{b_j}$ of $A_{N+1}\otimes B_y$, which in turn implies 
\begin{quote}
(**)~for all $y\in[N]$, $(\supp(A_{N+1})\otimes B_y)\ket{w}$ is the 0-vector.
\end{quote}
Write 
\begin{equation}\label{eq:sumtwoparts}
\sum_{x\in[N+1]} \bra{w}A_x\otimes B_x\ket{w} = \bra{w}A_{N+1}\otimes B_{N+1}\ket{w} + \sum_{x\in[N]} \bra{w}A_x\otimes B_x\ket{w}.
\end{equation}
Since $(A_{N+1}\otimes I)\ket{w_2}=0$ by definition of~$\ket{w_2}$, the first term on the right-hand side equals $\bra{w_1}A_{N+1}\otimes B_{N+1}\ket{w_1}$, which is $\leq\norm{w_1}^2$ by the base case.

For the second term, note that for all (not necessarily distinct) $x,y\in[N]$, we have 
$$
A_x\otimes B_y\ket{w_1}=(A_x\otimes B_y)(\supp(A_{N+1})\otimes I)\ket{w}=(A_x\otimes I)(\supp(A_{N+1})\otimes B_y)\ket{w},
$$ 
which is~0 because $(\supp(A_{N+1})\otimes B_y)\ket{w}=0$ by~(**).  
Thus we have $A_x\otimes B_y\ket{w}=A_x\otimes B_y\ket{w_2}$, which by~(*) also implies  that for all $x,y\in[N]$ with $x>y$ we have $\bra{w_2}A_x\otimes B_y\ket{w_2}=0$. Now the second term on the right-hand side of~(\ref{eq:sumtwoparts}) equals 
$$
\sum_{x\in[N]} \bra{w_2}A_x\otimes B_x\ket{w_2},
$$ 
which is $\leq\norm{w_2}^2$ by the induction hypothesis.
Since $\ket{w_1}$ and $\ket{w_2}$ are orthogonal, the two terms on the right-hand side of~(\ref{eq:sumtwoparts}) together are at most $\norm{w_1}^2 + \norm{w_2}^2 = \norm{w}^2$. This concludes the inductive step, and hence the proof of the claim.
\end{proof}

Applying Claim~\ref{claim:key} with the actual entangled state $\ket{\psi}$ used by protocol $Q$, we obtain
$$
N2^{-2q-1}\leq \sum_{x\in[N]}\Pr[\mbox{outcome }A_x\otimes B_x\mbox{ when measuring }\ket{\psi}] =\sum_{x\in[N]}\bra{\psi}A_x\otimes B_x\ket{\psi}\leq \norm{\psi}^2=1.
$$
Rearranging gives the theorem.
\end{proof}

\begin{corollary}
The $n$-bit equality, disjointness and greater-than functions have~$Q_1^*(f)\geq n/2-1/2$.
\end{corollary}

\begin{proof}
These three functions all have upper-triangular 1-fooling sets of size $2^n$.
\end{proof}

Now we use a trick of combining two copies of the function to extend the result from the equality function to all functions, at the expense of a factor of~2 in the lower bound (we do not know if this loss is necessary). This is similar to the proof that fooling set size is at most quadratically bigger than rank~\cite[Lemma~4.15]{kushilevitz&nisan:cc}:

\begin{corollary}\label{cor:generalfoolinglb}
For all $f:X\times Y\rightarrow\01$ we have $Q_1^*(f)\geq \frac{1}{4}\log \fool^1(f) - \frac{1}{2}$.
\end{corollary}

\begin{proof}
Define a new function $g:X^2\times Y^2\rightarrow\01$ by $g(xx',yy')=f(x,y)f(y',x')$. Note the reversed role of the two inputs in the second $f$.
 Alice and Bob can compute $g$ with one-sided error $p=1/4$ by separately computing $f(x,y)$ and $f^T(x',y')=f(y',x')$ with one-sided error $1/2$ each, and outputting the product of the two output bits.  This takes $Q^*_1(f)$ qubits of communication for each computation, so at most $2Q^*_1(f)$ in total.

Let $\{(x,x)\}$ be a 1-fooling set for $f$ of size $N=\fool^1(f)$.  Then it is easy to see that $\{(xx,xx)\}$ is a 1-fooling set for $g$, with the additional property that $g(xx,yy)=f(x,y)f(y,x)=0$ whenever $x\neq y$.  Hence the communication matrix for $g$ contains the $N\times N$ identity as a submatrix (i.e., the equality function).  The same proof as above gives a lower bound of $\frac{1}{2}\log N - 1$ for one-sided-error protocols for equality that accept 1-inputs with probability at least $1/4$ (instead of at least $1/2$ as above). Hence we have 
$\frac{1}{2}\log N - 1 \leq 2Q^*_1(f)$, which implies the statement.
\end{proof}

\section{Lower bound for nondeterministic quantum protocols}\label{secqndetLB}

In this section we study nondeterministic quantum protocols.
The following algebraic characterization of nondeterministic quantum communication complexity of $f$ is known.
The \emph{communication matrix} $M_f$ for $f$ is the $|X|\times|Y|$ Boolean matrix $M_f(x,y)=f(x,y)$.
A \emph{nondeterministic matrix} for $f$ is any real or complex matrix $M$ with the same support as $M_f$,
i.e., such that $M_{x,y}=0$ iff $f(x,y)=0$.
The \emph{nondeterministic rank} of $f$ (abbreviated to $\nrank(f)$) is the minimal rank (over the reals) among all such matrices.
\cite[Theorem~3.3]{wolf:nqj} shows that $NQ(f)=\ceil{\log\nrank(f)} + 1$.

The key to using fooling sets for nondeterministic quantum lower bounds is the following simple lemma:

\begin{lemma}
For every function $f:X\times Y\rightarrow\01$ we have $\nrank(f)^2\geq \fool^1(f)$.
\end{lemma}

\begin{proof}
Let $N= \fool^1(f)$.
Like in the proof of Corollary~\ref{cor:generalfoolinglb}, define
$g(xx',yy')=f(x,y)\cdot f(y',x')$ and observe that the communication matrix of $g$ contains the $N\times N$ identity matrix $I_N$ as a submatrix.
If $M$ is a nondeterministic matrix for $f$, then $M\otimes M^T$ is a nondeterministic matrix for $g$.
Hence, choosing $M$ of minimal rank, we have 
$$
\nrank(f)^2=\rank(M)^2=\rank(M\otimes M^T)\geq \nrank(g)\geq \nrank(I_N)=N.
$$
%
\end{proof}

Taking logarithms and using that $NQ(f)=\ceil{\log\nrank(f)}+1$, we get

\begin{corollary}\label{coro_nq}
$NQ(f)\geq\frac{1}{2}\log \fool^1(f) + 1$.
\end{corollary}

For example for the equality function, this shows $NQ(f)\geq n/2+1$.
However, for the equality function we already knew $NQ(f)=n+1$ since obviously $\nrank(f)=2^n$~\cite{wolf:nqj}.
Hence it is natural to ask whether the constant 1/2 in the above corollary is needed.
We don't know, but at least we can show that it needs to be less than~1.
Specifically, we give an example where $NQ(f)\leq\frac{\log 3}{\log 6}\log \fool^1(f) + 1$,
where $\frac{\log 3}{\log 6}\approx 0.613$.
Consider the following $6\times 6$ matrix:
$$
\left(
\begin{array}{rrrrrr}
    1 &    1 &    0 &    0 &    0 &    1\\
    0 &    1 &    0 &   -1 &   -1 &    0\\
   -1 &    1 &    1 &    0 &   -1 &    0\\
   -1 &    0 &    1 &    1 &    0 &    0\\
    1 &    0 &    0 &    1 &    1 &    1\\
    0 &    1 &    1 &    1 &    0 &    1
\end{array}
\right).
$$
It is easy to see that this has rank~3.
The Boolean matrix obtained by dropping the minus signs corresponds to a
communication complexity function $g:[6]\times[6]\rightarrow\01$ with a 1-fooling set of size~6
(just take the diagonal).  Now let $f:X\times Y\rightarrow\01$ be the AND of $k$ independent instances
of $g$ (so $|X|=|Y|=6^k$). Because 1-fooling set size is multiplicative under taking ANDs,
we have $\fool^1(f)=6^k$.  On the other hand, taking the $k$-fold tensor product of
the above rank-3 matrix gives a nondeterministic matrix for $f$ of rank $3^k$.
Hence $NQ(f)=\ceil{\log\nrank(f)}+1\leq \frac{\log 3}{\log 6}\log \fool^1(f)+1\approx 0.613\log \fool^1(f)$.

A simpler but slightly weaker separation can be obtained from the 3-input non-equality function, where $X=Y=[3]$ and the function take value~0 when the inputs $x$ and $y$ are equal.
This has $\nrank=2$ vs $\fool^1=3$, hence taking a $k$-fold AND of this gives a function
$f:X\times Y\rightarrow\01$ with $|X|=|Y|=3^k$ and $\nrank(f)=2^k$ vs $\fool^1(f)=3^k$.
Taking logarithms, we have $NQ(f)\approx 0.63\log \fool^1(f)$.

\section{Conclusion and open problems}

Equality and disjointness are two of the most important functions considered in communication complexity. Prior to this paper no large lower bound on the one-sided error or Las Vegas quantum communication complexity of these functions was known for the case of protocols with prior entanglement. In particular, for $\EQ$ previous lower bound methods gave only a constant lower bound. We have shown that the fooling set method is applicable to one-sided-error protocols with entanglement, obtaining linear lower bounds for both functions.

It is interesting to note that for classical protocols there is essentially no need to consider fooling sets at all: the method is completely subsumed by the rectangle bound (i.e., bounding the size of the largest monochromatic rectangle under some distribution). However, the rectangle bound does not apply to quantum protocols with one-sided error and entanglement, nor to quantum nondeterministic communication complexity, which is why considering fooling sets is important here.

We conclude with some open problems:
\begin{itemize}
\item Can we improve the factor~1/4 in Corollary~\ref{cor:generalfoolinglb}?  We believe it should be 1/2, which is what we already showed here for upper-triangular 1-fooling sets.
\item
Another problem is to show that the factor~1/2 in Corollary~\ref{coro_nq} is necessary. It seems hard to come up with a matrix for which the nondeterministic rank is the square root of the rank, as would be required by a construction along the lines of our separation at the end of Section~\ref{secqndetLB}.
\item
One further goal would be to show that classical deterministic complexity $D(f)$ and quantum Las Vegas complexity $Q_0(f)$ are polynomially close for all total functions. This is a (possibly easier) variant of a general conjecture that for total functions quantum communication yields only polynomial improvements in communication complexity.
Proving a linear lower bound in terms of classical nondeterministic complexity (i.e., $Q_0(f)=\Omega(N(f))$)
would settle that, since it is known that $D(f)=O(N(f)^2)$. However, an example from~\cite{wolf:nqj} refutes that hope.
Let $f(x,y)=0$ if $|x\wedge y|=1$ and $f(x,y)=1$ otherwise.
This function as well as its complement have linear $N(f)$, but $NQ(f),NQ(\neg f)=O(\sqrt{n})$.
This does not, however, preclude a bound like $Q_0(f)=\Omega(\sqrt{N(f)})$, which would still achieve the above goal.
\end{itemize}

\subsection*{Acknowledgements}
We thank Harry Buhrman and Matthias Christandl (as well as an anonymous referee) for pointing out an error in an earlier version of this paper, which we corrected here.

\bibliographystyle{alpha}

\newcommand{\etalchar}[1]{$^{#1}$}

\end{document}